\documentclass[a4paper,UKenglish]{lipics}
 
\usepackage{microtype}

\usepackage{tikz}
\usetikzlibrary{automata,arrows,positioning,calc,fit}

\newtheorem{claim}{Claim}

\title{Strategic Contention Resolution with Limited Feedback}

\author[1]{George Christodoulou}
\author[1]{Martin Gairing}
\author[2,3]{Sotiris Nikoletseas}
\author[2,3]{Christoforos Raptopoulos}
\author[1,2,3]{Paul Spirakis}

\affil[1]{Department of Computer Science, University of Liverpool, United Kingdom}
\affil[2]{Computer Engineering and Informatics Department, University of Patras, Greece}
\affil[3]{Computer Technology Institute \& Press ``Diophantus'', Greece
\texttt{G.Christodoulou@liverpool.ac.uk, gairing@liverpool.ac.uk, nikole@cti.gr} \\
\texttt{raptopox@ceid.upatras.gr, P.Spirakis@liverpool.ac.uk}}

\authorrunning{G. Christodoulou, M. Gairing, S. Nikoletseas, C. Raptopoulos, and P. Spirakis} 

\Copyright{G. Christodoulou, M. Gairing, S. Nikoletseas, C. Raptopoulos, and P. Spirakis}

\keywords{contention resolution, acknowledgment-based protocols, game theory}

\serieslogo{}

\begin{document}

\maketitle

\begin{abstract}
In this paper, we study contention resolution protocols from a game-theoretic perspective.
We focus on \emph{acknowledgment-based} protocols, where a user gets feedback from the channel only
when she attempts transmission. In this case she will learn whether her transmission was successful or not.
Users that do not transmit will not receive any feedback.
We are interested in equilibrium protocols, where no player has an incentive to deviate.

The limited feedback makes the design of equilibrium protocols a hard task as best response policies usually 
have to be modeled as Partially Observable Markov Decision Processes, which are hard to analyze.
Nevertheless, we show how to circumvent this for the case of two players and present an equilibrium protocol.
For many players, we give impossibility results for a large class of acknowledgment-based protocols, 
namely  \emph{age-based} and \emph{backoff} protocols with finite expected finishing time. Finally, we  
provide an age-based equilibrium protocol, which has infinite expected finishing time, but every player 
finishes in linear time with high probability.   
\end{abstract}

\section{Introduction}
\label{sec:intro}

{\em Contention resolution} in multiple access channels is one of the most 
fundamental problems in networking. In a multiple access channel (or broadcast channel) 
multiple users want to communicate with each other by sending messages into the channel. 
The channel is not centrally controlled, so two or more users can transmit their 
messages at the same time. If this happens then the messages collide and the transmission 
is unsuccessful. 
Contention resolution protocols specify how to resolve such conflicts, while simultaneously 
optimizing some performance measure, like channel utilization or average
throughput.

In this paper we follow the standard assumption that time is divided into discrete 
time slots, messages are broken up into fixed sized packets, and one packet fits exactly 
into one time slot. 
Moreover, we consider one of the simplest possible scenarios where there are
$n$ users, each of them having a single packet that needs to be
transmitted through the channel. 
When exactly one user
attempts transmission in a given slot, the transmission is
successful. However, if more than one users attempt transmission in
the same slot, a collision occurs, their transmission fails and they
need to retransmit their packages in later time slots.

Under centralized control of the users, avoiding collisions would be
simple: exactly one user would transmit at each time step,
alternating in a round-robin fashion. The complexity of the problem stems 
from the fact that there is no centralized control and therefore channel access has to
be managed by a distributed protocol. There is a large body of literature that studies
efficient distributed contention resolution protocols (see
Section~\ref{sec:related-work}). However, these protocols work under
the assumption that users will obediently follow the algorithm. In
this paper we follow~\cite{Fiat07} by dropping this assumption. 
We model
the situation as a non-cooperative stochastic game, where each user
acts as a selfish \emph{player} and tries to minimize the expected time before she
transmits successfully. Therefore a player will only obey a protocol if it
is in her best interest, given the other players stick to the protocol.

Fiat, Mansour, and Nadav~\cite{Fiat07} designed an incentive-compatible
transmission protocol which guarantees that (with high probability)
all players will transmit successfully in time linear in $n$. Their
protocol works for a very simple channel feedback structure, where
each player receives feedback of the form $0/1/2^+$ after each time
step ({\em ternary} feedback), indicating whether zero, one, or more
than one transmission was attempted.  Christodoulou, Ligett and Pyrga~\cite{CLP14}
designed equilibrium protocols for {\em multiplicity} feedback, where
each player receives as feedback the number of players that attempted
transmission\footnote{They also assume non-zero transmission costs, as
  opposed to \cite{Fiat07} and to this work.}.

The above protocols fall in the class of {\em full-sensing}
protocols~\cite{G03} where the channel feedback is broadcasted to all
sources. However, in wireless channels, there are situations where full-sensing 
is not possible because of the \emph{hidden-terminal problem} \cite{TK75}. 
In this paper, we focus on {\em acknowledgment-based} protocols, which use a 
more limited feedback model 
-- the only feedback that a user gets is whether her transmission was
successful or not. A user that does not transmit cannot ``listen'' to 
channel and therefore does not get any feedback. 
In other words, the only information that a user has is the history of 
her own transmission attempts.
Acknowledgment-based protocols have been extensively studied in the literature
(see e.g. \cite{G03} and references therein).
\emph{Age-based} and \emph{backoff} protocols both belong to the class of
acknowledgment-based protocols. 

Age-based protocols can be described by a sequence 
of probabilities (one for each time-step) of transmitting in each time step.
Those probabilities are given beforehand and do not change based on the 
transmission history. The well known ALOHA protocol~\cite{Abr70}  is a special 
age-based protocol, where -- except 
for the first round -- users always transmit with the same probability.
In contrast, in \emph{backoff} protocols, the probability of transmitting 
in the next time step only depends on the number of unsuccessful transmissions 
for the user. Here, a popular representative is the 
\emph{binary exponential backoff} mechanism, which is also used by the Ethernet protocol \cite{MB76}. 

The design
and the limitations of acknowledgment-based protocols is
well-understood 
\cite{140906,274816} if the users are not strategic. In this paper, we focus on the 
game-theoretic aspect of those protocols.

\subsection{Our Results}
\label{sec:contribution}

We study the design of acknowledgment-based {\em equilibrium
  protocols}. A user gets feedback only when she attempts transmission,
in which case she either receives an acknowledgment, in case of
success, or she realizes that a collision occurred (by the lack of an
acknowledgment). This model allows for very limited feedback, as
opposed to {\em full-sensing} protocols studied in \cite{Fiat07,CLP14}
where all players, even those who did not attempt transmission receive
channel feedback.

The feedback models used in \cite{Fiat07,CLP14} 
allow players, at each given time, to know exactly the number of
pending players. This information is very useful for the design of
equilibrium protocols. In our case, we assume that the number of pending players is
common knowledge only at the beginning. If a player chooses not to transmit
during a time-slot, then she is not sure how many players are still in the
game. From this time on, she can only sense the existence of other pending players 
when she participates in a collision.

The analysis of acknowledgment-based equilibrium protocols requires
different techniques. In full-sensing protocols, a best response for a
source can be modeled as an optimal policy of a Markov Decision
Process (MDP)~\cite{Fiat07}. For an acknowledgment-based protocol, this is in general
no longer possible, due to the uncertainty imposed by a
non-transmission. However, the best response policy in this case can be
modeled as a {\em Partially Observable Markov Decision Processes} (POMDP),
which are more complicated to analyze.

Lack of information makes the design of equilibrium protocols a hard
task. In particular, we show in Section \ref{sec:agedependent+backoff} 
that it is impossible to design an age-based or 
backoff protocol that is in equilibrium and has finite expected finishing time
\footnote{Note, that for more than two players, always transmitting is an equilibrium 
protocol with \emph{infinite} expected finishing time \cite{Fiat07}.}.
These impossibility results contribute to a partial characterisation of such protocols
and even hold for the case of two players.
This stands in contrast to the
full-sensing case for which the authors in~\cite{Fiat07} give an equilibrium protocol, 
where the $k$ remaining players 
transmit with probability 
$\Theta\left(\frac{1}{\sqrt{k}} \right)$.
This protocol finishes in finite but exponential time.

In Section \ref{sec:2playerprotocol}, we introduce and analyze an equilibrium 
protocol for two players. An interesting feature of our protocol is that each player is
using only limited information of her own history. More precisely, the probability of
transmission in a time-slot, depends only on whether a player
attempted transmission in the previous slot. 
Our proof reduces the POMDP for the best response policy 
to a \emph{finite} MDP, which we then analyze. This reduction crucially relies on the nature of our protocol.
We further show that our equilibrium protocol is the unique stationary equilibrium protocol.

For more than two players, we present an age-based equilibrium protocol. Although it has infinite expected finishing time,
every player finishes in linear time with high probability. Our protocol circumvents the lack of information by 
maintaining an estimation on the number of pending players, which with high probability is an upper bound on the actual number. 
The protocol uses a deadline mechanism similar to~\cite{Fiat07}. 
Their protocol exploits the existence of their finite time equilibrium protocol 
mentioned above. For our more restricted model it is not known if such a finite time protocol exists 
for more than two players. 
This is the main open question left from our work.
We stress that
our negative results exclude the possibility that such a protocol can be age-based or backoff.

All missing proofs are included in a clearly marked Appendix.

 \subsection{Related Work}
 \label{sec:related-work}

 The ALOHA protocol, introduced by Abramson \cite{Abr70} 
 (and modified by Roberts~\cite{Rob72} to its slotted version), 
 is a multiple-access communication protocol, which has been around since 
 the 70's.
 Many subsequent papers study the efficiency of multiple-access protocols
 when packets are generated by some stochastic process (see for
 example \cite{355567,310333,903752}), while worst-case scenarios of
 bursty inputs, were studied in \cite{1074023}. To model such a
 worst-case scenario, one needs $n$ nodes, each of which must
 simultaneously transmit a packet; this is also the model we use in
 this work.
 
 A large class of contention resolution protocols explicitly deals
 with {\em conflict resolution}; where if $k\geq 2$ users collide (out of
 a total of $n$ users), then a resolution algorithm is called on to
 resolve this conflict (by ensuring that all the colliding packets are
 successfully transmitted), before any other source is allowed to use
 the
 channel~\cite{capetanakis-a,capetanakis-b,hayes78,TsM}. 
%
 There have been many positive and negative results on the efficiency
 of protocols under various information models (see \cite{G03} for an
 overview of results). When $k$ is known,~\cite{140906} provides an
 $O(k + \log k \log n)$ {\em acknowledgment-based} algorithm,
 while~\cite{274816} provides a matching lower bound. For the ternary
 model,~\cite{GW85} provides a bound of $\Omega(k(\log n/\log k))$ for
 all deterministic algorithms.
%
 
 A variety of game theoretic models of slotted ALOHA have also
 been proposed and studied; see for
 example~\cite{1031826,1154073,Altman05azouzi}.
 However, much of this work only considers transmission protocols that
 always transmit with the same fixed probability (perhaps as a
 function of the number of players in the game).  
 Other game theoretic approaches have considered pricing
 schemes~\cite{1285895} and cases in which the channel quality changes
 with time and players must choose their transmission levels
 accordingly~\cite{Menache08,1288109,AMPP08}. \cite{KoutsoupiasP12}
 studied a game-thoretic model that lies between the contention and
 congestion model, where the decision of {\em when} to submit is part
 of the action space of the players. As discussed in the previous
 section, the most relevant game-theoretic model to our work, is the
 one studied by  Fiat, Mansour, and Nadav~\cite{Fiat07} and by  Christodoulou, Ligett, and Pyrga~\cite{CLP14}. In \cite{CLP14}, efficient $\epsilon$-equilibrium
 protocols are designed, but the authors assume non-zero transmission costs, in
 which case the efficient protocol of \cite{Fiat07} does not
 apply. Their protocols use {\em multiplicity} feedback (the number of
 attempted transmissions) which again falls in the class of
 full-sensing protocols.



\section{Model}
\label{sec:definitions}

\noindent \textbf{Game Structure.} Let $N=\{1, 2, \dots, n\}$ be the
set of agents, each one of which has a single packet that he wants to
send through a common channel. All players know $n$. 
We assume time is discretized into
slots $t = 1, 2, \ldots$. The players that have not yet successfully
transmitted their packet are called \emph{pending} and initially all
$n$ players are pending. At any given time slot $t$, a pending player
$i$ has two available actions, either to \emph{transmit} his packet or
to \emph{remain quiet}.
%
%
In a \emph{(mixed) strategy}, a player $i$ transmits his packet at
time $t$ with some probability that potentially depends on information
that $i$ has gained from the channel based on previous transmission
attempts. If exactly one player transmits in a given slot $t$, then
his transmission is \emph{successful}, the successful player exits the
game (i.e. he is no longer pending), and the game continues with the
rest of the players. On the other hand, whenever two or more agents
try to access the channel (i.e. transmit) at the same slot, a
\emph{collision} occurs and their transmissions fail, in which case
the agents remain in the game. Therefore, in case of collision or if
the channel is idle (i.e. no player attempts to transmit) the set of
pending agents remains unchanged. The game continues until all players
have successfully transmitted their packets.


\noindent \textbf{Transmission protocols.}  Let $X_{i, t}$ be the
indicator variable that indicates whether player $i$ attempted
transmission at time $t$. For any $t \geq 1$, we denote by $\vec{X}_t$
the transmission vector at time $t$, i.e. $\vec{X}_t = (X_{1, t},
X_{2, t}, \ldots, X_{n, t})$. 
An {\em acknowlegment-based} protocol, uses very limited channel
feedback. After each time step $t$, only players that attempted a
transmission receive feedback, and the rest get no information. In
fact, the information received by a player $i$ who transmitted during
$t$ is whether his transmission was successful (in which case he gets
an acknowledgement and exits the game) or whether there was a
collision.

Let $\vec{h}_{i, t}$ be the vector of the \emph{personal transmission
  history} of player $i$ up to time $t$, i.e. $\vec{h}_{i, t} = (X_{i,
  1}, X_{i, 2}, \ldots, X_{i, t})$. We also denote by $\vec{h}_t$ the
transmission history of all players up to time $t$, i.e. $\vec{h}_t =
(\vec{h}_{1, t}, \vec{h}_{2, t}, \ldots, \vec{h}_{n, t})$. In an
acknowledgement-based protocol, the actions of player $i$ at time $t$
depend only (a) on his personal history $\vec{h}_{i, t-1}$ and (b) on
whether he is pending or not at $t$. A \emph{decision rule} $f_{i, t}$
for a pending player $i$ at time $t$, is a function that maps
$\vec{h}_{i, t-1}$ to a probability $\Pr(X_{i, t} = 1 | \vec{h}_{i,
  t-1})$. For a player $i \in N$, a \emph{(transmission) protocol}
$f_i$ is a sequence of decision rules $f_i = \{ f_{i, t}\}_{t \geq 1}
= f_{i, 1}, f_{i, 2}, \cdots$.

A transmission protocol is \emph{anonymous} if and only if the
decision rule assigns the same transmission probability to all players
with the same personal history. In particular, for any two players $i
\neq j$ and any $t \geq 0$, if $\vec{h}_{i, t-1} = \vec{h}_{j, t-1}$,
it holds that $f_{i, t}(\vec{h}_{i, t-1}) = f_{j, t}(\vec{h}_{j,
  t-1})$. In this case, we drop the subscript $i$ in the notation,
i.e. we write $f = f_1 = \cdots = f_n$.

We call a protocol $f_i$ for player $i$ \emph{age-based} if and only
if, for any $t \geq 1$, the transmission probability $\Pr(X_{i, t} = 1
| \vec{h}_{i, t-1})$ depends only (a) on time $t$ and (b) on whether
player $i$ is pending or not at $t$. In this case, we will denote the
transmission probability by $p_{i, t} \stackrel{def}{=} \Pr(X_{i, t} =
1 | \vec{h}_{i, t-1}) = f_{i, t}(\vec{h}_{i, t-1})$.

A protocol is called \emph{backoff} if the decision rule at time $t$
is a function of the number of {\em unsuccessful} transmissions. 
We call a transmission protocol $f_i$ \emph{non-blocking} if and only
if, for any $t \geq 1$ and any transition history $\vec{h}_{i, t-1}$,
the transmission probability $\Pr(X_{i, t} = 1 | \vec{h}_{i, t-1})$ is
always smaller than
1. 
A protocol $f_i$ for player $i$ is a \emph{deadline protocol with
  deadline} $t_0 \in \{1, 2, \ldots\}$ if and only if $f_{i,
  t}(\vec{h}_{i, t-1}) = 1$, for any player $i$, any time slot $t \geq
t_0$ and any transmission history $\vec{h}_{i, t-1}$. A
\emph{persistent player} is one that uses the deadline protocol with
deadline $1$.

\noindent \textbf{Efficiency.} Assume that all $n$ players in the game
employ an anonymous protocol $f$. We will say that $f$ is \emph{efficient} if and only if all players will have successfully transmitted by time $\Theta(n)$ with high probability (i.e. with probability tending to 1, as $n$ goes to infinity). 

\noindent \textbf{Individual utility.} Let ${\vec f} = (f_1, f_2,
\ldots, f_n)$ be such that player $i$ uses protocol $f_i, i \in
N$. For a given transmission sequence $\vec{X}_1, \vec{X}_2, \ldots$,
which is consistent with ${\vec f}$, define the \emph{latency} or
\emph{success time} of agent $i$ as $T_i \stackrel{def}{=}
\inf\{t:X_{i, t} = 1, X_{j, t} = 0, ~ \forall j \neq i\}$. That is,
$T_i$ is the time at which $i$ successfully transmits.  Given
a 
transmission history $\vec{h}_t$, 
the $n$-tuple of protocols ${\vec f}$ induces a probability
distribution over sequences of further transmissions. In that case, we
write $C^{{\vec f}}_i(\vec{h}_t) \stackrel{def}{=}
\mathbb{E}[T_i|\vec{h_t}, {\vec f}] = \mathbb{E}[T_i|\vec{h}_{i, t},
{\vec f}]$ for the expected latency of agent $i$ incurred by a
sequence of transmissions that starts with $\vec{h}_t$ and then
continues based on ${\vec f}$. For anonymous protocols, i.e. when $f_1
= f_2 = \cdots = f_n = f$, we will simply write $C^{f}_i(\vec{h}_t)$
instead\footnote{Abusing notation slightly, we will also write
  $C^{{\vec f}}_i(\vec{h}_0)$ for the \emph{unconditional} expected
  latency of player $i$ induced by ${\vec f}$.
}.

\noindent \textbf{Equilibria.} The objective of every agent is to
minimize her expected latency. We say that ${\vec f} = \{f_1, f_2,
\ldots, f_n \}$ is in \emph{equilibrium} if for any transmission
history $\vec{h}_t$ the agents cannot decrease their expected latency
by unilaterally deviating after $t$; that is, for all agents $i$, for
all time slots $t$, and for all decision rules $f'_i$ for agent $i$,
we have
\begin{equation*}
  C^{{\vec f}}_i(\vec{h}_t) \leq C^{({\vec f}_{-i}, f'_i)}_i(\vec{h}_t),
\end{equation*}
where $({\vec f}_{-i}, f'_i)$ denotes the protocol
profile\footnote{For an anonymous protocol $f$, we denote by $(f_{-i}, f'_i)$ the
profile where agent $j \neq i$ uses protocol $f$ and agent $i$
uses protocol $f'_i$.} where every
agent $j \neq i$ uses protocol $f_j$ and agent $i$ uses protocol
$f'_i$. 





\section{An equilibrium protocol for two
  players} \label{sec:2playerprotocol}

In this section we show that there is an anonymous
acknowledgment-based protocol in equilibrium, when
$n=2$. 

We define the protocol $f$ as follows: for any $t\geq 1$, player $i$ and
transmission history $\vec{h}_{i, t-1}$,
\begin{equation} \label{2-player-equilibrium}
f_{i, t}(\vec{h}_{i, t-1}) = \left\{
\begin{array}{ll}
	\frac{2}{3}, & \quad \textrm{if $X_{i,t-1}=1$ or $t=1$} \\
	1, & \quad \textrm{if $X_{i,t-1}=0$.}	 
\end{array}
\right.
\end{equation} 

\begin{theorem}
There is an anonymous acknowledgment-based equilibrium protocol for two players.
\end{theorem}

\begin{proof}

  We will show that protocol $f$ is in equilibrium.  Let Alice and Bob
  be the two players in the system. We will show that when Bob sticks
  with playing $f$, any deviation for Alice, at any possible slot,
  will be less profitable for her. 

  Let's denote by $C^{f,j}_i$, for $j\in\{0,1\}$, the expected
success time for a pending player $i$ given that in the last round he attempted
transmission $(j=1)$ or not $(j=0)$ i.e., $C^{f,j}_i =
\mathbb{E}[T_i|\vec{h_t}, f,X_{i,t}=j]$. The following claim
asserts that the expected success time for Alice depends only on
whether she attempted a transmission or not in the previous slot. For the proof, we compute the expected time to absorption for the Markov chain ${\cal M}$ shown in Figure~\ref{MC-both-f}, starting from states $A$ and $B$. The full details can be found in Appendix \ref{appendix:claim:protocol}. 

\begin{claim}\label{claim:protocol}
$C^{f,j}_{Alice}=2+j$, for $j\in\{0,1\}$.
\end{claim}

  \begin{figure}[!ht]
  \centering
    \subfloat[Markov chain ${\cal M}$.\label{MC-both-f}]{%
\resizebox{.4\textwidth}{!}{%
\begin{tikzpicture}[->, >=stealth', auto, semithick, node distance=4cm]
\tikzstyle{every state}=[fill=white,draw=black,thick,text=black,scale=1.1]
\node[state]    (A)                     {$A$};
\node[state]    (B)[right of=A]         {$B$};
\node[state]    (C)[below of=A]         {$C$};
\node[state]    (D)[below of=B]         {$D$};
\path
(A) edge[loop left]        node{$\frac{4}{9}$}   (A)
    edge[bend left]        node{$\frac{1}{9}$}   (B)
    edge                   node{$\frac{2}{9}$}   (C)
    edge                   node{$\frac{2}{9}$}   (D)
(B) edge[bend left, above] node{$1$}   (A)
(C) edge                   node{$1$}                (D);
\draw[rotate around={135:(2.2,-2.1)},red, dashed]  ($(2.2,-2.1)$) ellipse (1.1cm and 4cm)node[yshift=0cm]{};

\end{tikzpicture}
}
    }
    \hfill
    \subfloat[Markov chain ${\cal M}'$.\label{MC-Alice-deviates}]{%
\resizebox{.47\textwidth}{!}{%
\begin{tikzpicture}[->, >=stealth', auto, semithick, node distance=4cm]
\tikzstyle{every state}=[fill=white,draw=black,thick,text=black,scale=1.1]
\node[state]    (A)                     {$A$};
\node[state]    (E)[right of=A]         {$E$};
\node[state]    (F)[below of=A]         {$F$};
\node[state]    (D)[below of=E]         {$D$};
\path
(A) edge[loop left]        node{$\frac{2}{3}p_A$}   (A)
    edge[bend left]        node{$1-p_A$}            (E)
    edge                   node[,below,pos=0.2]{$\frac{1}{3}p_A$}   (D)
(E) edge[bend left, above] node{$\frac{1}{3}p_E$}   (A)
    edge                   node{$\frac{2}{3}p_E$ \phantom{space}}   (D)
    edge                   node[pos=0.28]{$1-p_E$}  (F)
(F) edge                   node{$1$}                (D);
\end{tikzpicture}
}
    }
    \caption{Markov chains used in the analysis.}
    \label{fig:MCs}
  \end{figure}
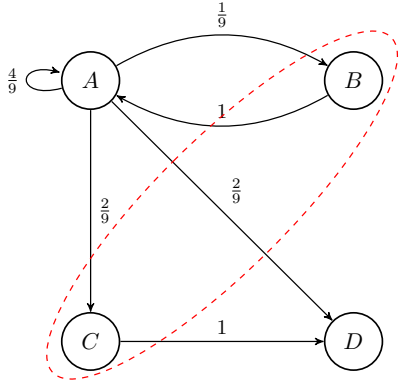
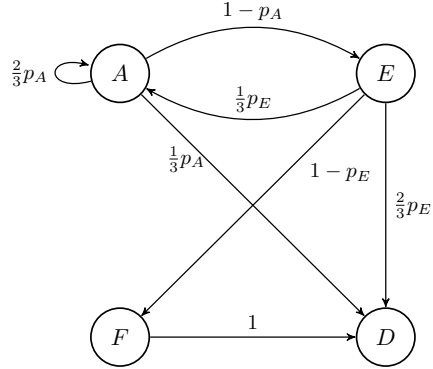





It remains to be shown that for any transmission history up to any time $t$, the optimal (best-response) strategy for
Alice is to follow $f$.
Notice that this situation from Alice's point of view can be described
by an infinite-horizon, undiscounted \emph{Partially observable Markov
  Decision Process (POMDP)}, by the direct modification of the Markov
chain ${\cal M}$ that is described in the proof of
Claim~\ref{claim:protocol}. This process is partially observable due
to the uncertainty created whenever Alice does not attempt
transmission. This creates complications in the analysis, as general
results about the existence of optimal stationary policies in MDPs
\cite{MDPbook}, do not carry over immediately and also optimal
policies are not always well-defined for undiscounted POMDPs with
infinite horizon~\cite{Pl80}. Fortunately, by exploiting the nature of
our specific protocol $f$, and in particular the fact that a player
using $f$ never misses two transmissions in a row, we are able to
circumvent this difficulty and model the situation as an MDP.

Following the notation in \cite{Nbook}, the \emph{state space} of the
MDP is ${\cal I} = \{A, E, F, D\}$. The states are interpreted as
follows: As in the Markov chain ${\cal M}$, state $A$ 
describes the situation in which both players are pending and they both
know it (this is reached just after a collision, or at time $t=1$) and state $D$
corresponds to the state in which Alice successfully transmitted. $F$ is the
state in which Alice 
did not transmit for two consecutive rounds. Since Bob follows $f$, he will have 
transmitted in one of these two rounds. Thus, in $F$ Alice 
is the only pending player and she knows it. Note that in $F$ the unique optimal strategy 
for Alice is to transmit in the next round.
Finally, $E$ is the state in which Alice is
uncertain whether she is the only pending player in the system; this
happens at $t$ if she did not transmit at $t-1$, but transmitted at
$t-2$ and there was a collision. State $E$ essentially corresponds to a combination of 
states $B$ and $C$ in Figure \ref{MC-both-f}.

Since Alice clearly starts at state $A$, the \emph{initial
  distribution} of the MDP is $\lambda$, where $\lambda_A = 1$ and
$\lambda_E = \lambda_F = \lambda_D = 0$. The set of \emph{actions} for
Alice is ${\cal A} = [0, 1]$. In particular, if Alice decides to take
action $a \in {\cal A}$ at time $t$, then she will transmit with probability
$a$ at $t$. Furthermore, the \emph{cost function} of the MDP is $c(a)
= (c_s(a): s \in {\cal I})$ and we have $c_A(a) = c_E(a) = c_F(a) = 1$
and $c_D(a) = 0$ for all $a \in {\cal A}$. Finally, for the
\emph{transition matrix} of our MDP, notice that, since the MDP
describes the situation from Alice's perspective, we calculate
transition probabilities by ``deferring'' the relevant decisions taken
by Bob until the time that Alice gets feedback
. The transition matrix of our MDP is shown in equation
(\ref{MDP-transition-matrix}) and it is explained in more detailed
below.
\begin{equation} \label{MDP-transition-matrix}
P(a) = \left[ 
\begin{array}{cccc}
	\frac{2a}{3} & 1-a & 0 & \frac{a}{3} \\
	\frac{a}{3} & 0 & 1-a & \frac{2a}{3} \\
	0 & 0 & 1-a & a \\
	0 & 0 & 0 & 1
\end{array}
\right].
\end{equation}
In particular, we can see from
(\ref{MDP-transition-matrix}) 
that the probability to visit state $A$ in one step, given that we are
at state $E$ and the action taken is $a \in [0, 1]$, is $P_{E, A}(a) =
\frac{a}{3}$. Indeed, this happens at some time $t$ if at time $t-1$
Alice did not transmit but Bob did not transmit either; therefore, by
definition of $f$, given that we are at $E$ (i.e. Alice did not
transmit at time $t-1$), the probability that we reach $A$ is equal to
the probability that Alice transmits at $t$ (which happens with
probability $a$) multiplied by the probability that Bob did not
transmit at $t-1$ (which happens with probability
$\frac{1}{3}$). Similarly, the probability that we visit $D$ in one
step, given that we are at state $E$ and the action taken is $a \in
[0, 1]$, is $P_{E, D}(a) = \frac{2a}{3}$, which is the probability
that Alice transmits in the current step and Bob transmitted in the
previous one (in which Alice did not transmit, thus Bob was
successful).

By Lemma 5.4.2 and Theorem 5.4.3 from \cite{Nbook}, there is a
stationary policy (i.e. protocol) $u^*$ that is optimal in the sense
that it achieves the minimum expected total cost, given that we start
at state $A$. The fact that $u^*$ is stationary significantly reduces
the search space of optimal strategies. In particular, this allows us
to only consider strategies for which the actions taken by Alice (in
the above MDP) depend only on the current state. In fact, we can
further reduce the family of optimal strategies considered by noting
that in any optimal strategy Alice will transmit with probability 1
when in state $F$; indeed, when Alice knows that she is the only
pending player, she will decide to transmit with probability 1 in the
next time step. Therefore, it only remains to determine the
probability of transmission when we are at either state $A$ of $E$;
denote those by $p_A$ and $p_E$ respectively. Therefore, this leads to
a Markov chain ${\cal M'}$ with state space ${\cal I}' = {\cal I}$ and
transition probabilities that correspond to actions from the above
MDP. The transition graph of ${\cal M'}$ is shown in Figure
\ref{MC-Alice-deviates}.


Clearly, the expected latency of Alice when she uses protocol $u^*$
and Bob uses protocol $f$ is equal to the expected hitting time
$k_{A}^{'D}$ that ${\cal M}'$ needs to reach state $D$, given that we
start from $A$. By definition, we have $k_{F}^{'D} = 1, k_{D}^{'D} =
0$, and by the Markov property, we get $k_{A}^{'D} = 1 + \frac{2}{3} p_A k_{A}^{'D} + (1-p_A) k_{E}^{'D}$ and $k_{E}^{'D} = 1 + \frac{1}{3} p_E k_{A}^{'D} + (1-p_E) k_{F}^{'D}$.
Rearranging and after substitutions we get $k_{A}^{'D} =
3$ and $k_{E}^{'D} = 2$, for any $p_A, p_E \in [0, 1]$. Comparing this to Claim \ref{claim:protocol}, 
we conclude that if Bob uses  $f$, a best response for Alice is to also follow $f$. 
This completes the proof of the Theorem. \end{proof}





\subsection{Uniqueness}
\label{sec:uniqueness}

We will say that a protocol is \emph{stationary} if the decision rule for each player at some time $t$ depends on the information state of the player at $t$. In particular, the protocol defined in equation (\ref{2-player-equilibrium}) is stationary. In this section we show that there are no other stationary equilibria. 

\begin{theorem} 
For two players, the unique stationary anonymous protocol that is in equilibrium is the one defined in equation (\ref{2-player-equilibrium}).
\end{theorem}
\proof For the sake of contradiction, assume that there is another stationary protocol that is in equilibrium. As in the analysis of protocol (\ref{2-player-equilibrium}) in Section \ref{sec:2playerprotocol}, we denote by $A$ the state where both players know they are both pending. Let Alice be one of the two players. Notice that, every time Alice transmits, either there is a collision (in which case Alice returns to state $A$) or the transmission is successful (so Alice is no longer pending).

For $k=1, 2, \ldots,$ let $p_k$ denote the probability that Alice transmits in step $k$, given that she starts from $A$ at $t=0$ and she does not transmit in time steps $1$ to $k-1$. Therefore, given that we start from $A$ at time 0, the probability that Alice attempts to transmit for the first time after $k$ steps is $p_{k}\prod_{k'=1}^{k-1} (1-p_{k'})$. In particular, in the equilibrium described in the previous section, we had
$p_1=\frac{2}{3}$ and $p_2=1$.



%
First, assume there is another stationary protocol $g$ that is in equilibrium, for which $p_2=1$ and $p_1=p\neq \frac{2}{3}$. 
Adjusting the transition probabilities in the Markov chain in Figure \ref{MC-both-f} accordingly,
and doing the same analysis we can derive that the expected latency of Alice when both players use protocol $g$ is 
$k_A^D=\frac{2-p}{2p(1-p)}$. We will show that for all $p\neq \frac{2}{3}$ a player has a profitable
deviation. Indeed, first observe that $p>\frac{2}{3}$ implies $k_A^D>3$. In this case Alice can improve her expected latency by not transmitting 
for two consecutive time steps and then (successfully) transmitting in the third time step. Second, for the case that $p<\frac{2}{3}$, persistently 
transmitting in each time step is a deviation which gives the deviator an expected latency of $\frac{1}{1-p}$.
For $p<\frac{2}{3}$ this is strictly less than the expected latency $k_A^D=\frac{2-p}{2p(1-p)}$ that Alice has when both players use protocol $g$.
From both cases, we conclude that there is no stationary protocol in equilibrium for which $p_2=1$ and $p_1\neq \frac{2}{3}$. 

Now assume that there is another stationary protocol $z$ in equilibrium, for which $p_2<1$. Denote $\alpha_z$ the expected 
latency of Alice when both players use protocol $z$. Similarly denote $\alpha_{(z')}$ the expected latency of Alice when she unilaterally deviates from $z$ to some other protocol $z'$. 
We will consider the following three protocols that Alice can use instead of $z$: (i) Using protocol $(1z)$, Alice will transmit in the first time step and then continue by following protocol $z$. 
(ii) Using protocol $(01z)$, Alice will not transmit in the first time step, but will transmit in the second time step and then follow the protocol $z$. 
(iii) Finally, using protocol $(001z)$, Alice will not transmit for the first two time steps, but will transmit in the third time step and then follow the protocol $z$.
The expected latency of Alice when she uses each of those protocols while the other player uses $z$ is given by:
\begin{eqnarray*}
\alpha_{(1z)}  & = & 1 + p_1 \alpha_z \\
\alpha_{(01z)} & = & 2 + (1-p_1)p_2 \alpha_z \\
\alpha_{(001z)} & = & 3 + (1-p_1)(1-p_2)p_3 \alpha_z. 
\end{eqnarray*}
Notice now that all three transmission sequences $(1), (0, 1)$ and $(0, 0, 1)$ are consistent with $z$. Furthermore, $z$ is acknowledgment-based, so Lemma \ref{lemma:bestresponses} applies here. Therefore, the above expected latencies must all be equal to $\alpha_z$.
Using the identities $\alpha_z=\alpha_{(1z)}=\alpha_{(01z)}$ we get that 
$\alpha_z=2+p_2 < 3$. But clearly $3 \leq \alpha_{(001z)}$, which is a contradiction to the fact that $\alpha_z=\alpha_{(001z)}$. 
Thus, there is no equilibrium protocol with $p_2<1$. This completes the proof of the theorem. \qed




\section{Age-based and backoff protocols}
\label{sec:agedependent+backoff}

In this section, we focus on two special prominent classes of
acknowledgment-based protocols, namely {\em age-based} and {\em
  backoff}, and we show that these cannot be implemented in
equilibrium if we insist on finite expected latency.

In what follows, for any protocol $f$, any player $i$ that uses $f$
and any time $t$, we will say that $\vec{h}_{i, t}$ is
\emph{consistent with $f$} if and only if there is a non-zero
probability that $\vec{h}_{i, t}$ will occur for player $i$. 

Now we are ready to show in the next Lemma a useful property of all
\emph{acknowledgment-based} equilibrium
protocols 
that is essentially an analogue of the property of Nash equilibria for
finite games that all pure strategies in the support of a Nash
equilibrium are best responses.

\begin{lemma} \label{lemma:bestresponses} Let $f \stackrel{def}{=}
  \{f_t\}_{t \geq 1}$ be an anonymous acknowledgment-based protocol
  and let $\pi \stackrel{def}{=} \pi_1, \pi_2, \ldots$ be any 0-1
  sequence which is consistent with $f$. For any (finite) positive
  integer $\tau^*$, define the protocol
  \begin{equation} \label{eq-deviator} 
	g = g(\tau^*) \stackrel{def}{=}
    \left\{
\begin{array}{ll}
	\pi_t, & \quad \textrm{for $1 \leq t \leq \tau^*$} \\
	f_t, & \quad \textrm{for $t > \tau^*$.}
      \end{array}
 \right.
\end{equation}
We then have that, for any fixed player $i$, if $f$ is in equilibrium,
then
\begin{equation*}
C^f_i(\vec{h}_0) = C^{(f_{-i}, g)}_i(\vec{h}_0).
\end{equation*}
\end{lemma}
\proof Since we consider acknowledgment-based protocols, for the sake
of the analysis, we will assume that players continue to flip coins
even after successfully transmitting, so that they eventually find out
what their decisions would have been at any time $t$.\footnote{In
  fact, we only need this assumption to hold for any $t$ which is at
  most some predefined fixed upper bound $\tau^*$.}

For a fixed player $i$, we obtain
\begin{equation} \label{eq-expectation}
C^f_i(\vec{h}_0) = \mathbb{E}[T_i| \vec{h}_{i, 0}, f] = \sum_{\vec{h}_{i, \tau^*}} \mathbb{E}[T_i| \vec{h}_{i, \tau^*}, f] \Pr\left\{\textrm{$\vec{h}_{i, \tau^*}$ happens for $i$} \right\}.
\end{equation}
Notice now that, since $f$ is acknowledgment-based, the event
$\left\{\textrm{$\vec{h}_{i, \tau^*}$ happens for $i$} \right\}$ is
independent of the transmission sequences of other players. Therefore,
$\mathbb{E}[T_i| \vec{h}_{i, \tau^*}, f]$ is equal to the
unconditional (i.e. conditioned on $\vec{h}_{i, 0}$) expected latency
of player $i$ when she uses the protocol defined in equation
(\ref{eq-deviator}), where the first $\tau^*$ terms of $\pi$ are
replaced by $(\pi_1, \ldots, \pi_{\tau^*}) = (\vec{X}_{i, 1}, \ldots,
\vec{X}_{i, \tau^*}) = \vec{h}_{i,
  \tau^*}$.\footnote{Note that this observation is not true for
  general protocols and different kinds of feedback, which is why the
  present analysis cannot be used to prove an impossibility result in
  the case of protocols like those in \cite{Fiat07}.} In particular,
we have that $\mathbb{E}[T_i| \vec{h}_{i, \tau^*}, f] =
\mathbb{E}[T_i| \vec{h}_{i, 0}, (f_{-i}, g)] = C^{(f_{-i},
  g)}_i(\vec{h}_0)$.

Assume now for the sake of contradiction that there is a transmission
history $\vec{h}_{i, \tau^*}$ for player $i$ such that
$\mathbb{E}[T_i| \vec{h}_{i, \tau^*}, f] \neq \mathbb{E}[T_i|
\vec{h}_{i, 0}, f]$. Clearly, if $\mathbb{E}[T_i| \vec{h}_{i, \tau^*},
f] < \mathbb{E}[T_i| \vec{h}_{i, 0}, f]$, then the protocol $g$ is a
better protocol for player $i$, which contradicts the fact that $f$ is
in equilibrium. On the other hand, if $\mathbb{E}[T_i| \vec{h}_{i,
  \tau^*}, f] > \mathbb{E}[T_i| \vec{h}_{i, 0}, f]$, then equation
(\ref{eq-expectation}) implies that there must be another transmission
history $\vec{h}'_{i, \tau^*}$ for which $\mathbb{E}[T_i| \vec{h}'_{i,
  \tau^*}, f] < \mathbb{E}[T_i| \vec{h}_{i, 0}, f]$. 

Therefore, we have that $C^{(f_{-i}, g)}_i(\vec{h}_0) =
\mathbb{E}[T_i| \vec{h}_{i, 0}, (f_{-i}, g)] = \mathbb{E}[T_i|
\vec{h}_{i, \tau^*}, f] = \mathbb{E}[T_i| \vec{h}_{i, 0}, f] =
C^f_i(\vec{h}_0)$, for any transmission history $\vec{h}_{i, \tau^*}$,
and for any finite $\tau^* \geq 1$, thus also for any 0-1 sequence
$\pi$ that is consistent with $f$. \qed

\medskip

The next corollary is an interesting consequence of
Lemma~\ref{lemma:bestresponses} regarding {\em non-blocking} anonymous
age-based protocols.
The full proof can be found in Appendix \ref{appendix:cor:nonblocking}.

\begin{corollary} \label{cor:nonblocking} 
Let $f \stackrel{def}{=}
  \{f_t\}_{t \geq 1}$ be a non-blocking anonymous age-based
  protocol. If the expected latency of a player using protocol $f$ is
  finite, i.e. $\mathbb{E}[T_i| \vec{h}_{i, 0}, f] < \infty$, then $f$
  is not in equilibrium.
\end{corollary}

We are now ready to show the main result of this section.

\begin{theorem} \label{lemma:nooblivious} There is no anonymous
  age-based protocol $f$ for $n \geq 2$ players that is in equilibrium
  and has $\mathbb{E}[T_i| \vec{h}_{i, 0}, f] < \infty$, for any
  player $i$.
\end{theorem}
\proof For the sake of contradiction, let's assume that $f =
\{f_t\}_{t \geq 1}$ is an age-based protocol in equilibrium with
finite expected latency, i.e. $\mathbb{E}[T_i| \vec{h}_0, f] <
\infty$. 
The next claim asserts the existence of a finite positive integer
$\tau^*$ where the protocol dictates transmission, with certain
properties, which will be a useful ingredient for the rest of the proof. The detailed proof of Claim \ref{claim:agebased} can be found in Appendix \ref{appendix:claim:agebased}.

\begin{claim} \label{claim:agebased} Let $f$ be an anonymous age-based
  protocol for $n$ players that is in equilibrium and has
  $\mathbb{E}[T_i| \vec{h}_{i, 0}, f] < \infty$, then there is a
  finite positive integer $\tau^*$ such that

\begin{description}

\item[(a)] $f_{\tau^*} = 1$, 

\item[(b)] $f_{\tau^*-1} < 1$ and 

\item[(c)] there exist $\tau_1 < \cdots < \tau_{n-1} < \tau^*$, such that $f_{\tau_j} < 1$, for all $j = 1, \ldots, n-1$.

\end{description}
\end{claim}

Take a $\tau^*$ as described in the above claim and consider the
protocol $Q$ defined as follows
\begin{equation} 
Q \stackrel{def}{=} \left\{
\begin{array}{ll}
	0, & \quad \textrm{if $f_t < 1$, for $1 \leq t \leq \tau^*-2$} \\
	1, & \quad \textrm{if $f_t = 1$, for $1 \leq t \leq \tau^*-2$} \\
	1, & \quad \textrm{for $t = \tau^*-1$ and $t = \tau^*$} \\
	f_t, & \quad \textrm{for $t > \tau^*$.}
\end{array}
 \right.
\end{equation}
Notice that, since the initial (deterministic) sequence of transmissions of $Q$ is consistent with $f$, by Lemma
\ref{lemma:bestresponses} we have that $\mathbb{E}[T_i| \vec{h}_{i,
  0}, f] = \mathbb{E}[T_{i}|\vec{h}_{i, 0}, (f_{-i}, Q)]$.

Now consider the protocol $Q'$, which is the same as $Q$, with the
only difference\footnote{Note that $Q'$ does not agree with $f$
  whenever $f_t = 1$, so Lemma \ref{lemma:bestresponses} does not
  apply to $Q'$.} that $Q'_{\tau^*} = 0$. In fact, we show that,
$\mathbb{E}[T_{i}|\vec{h}_{i, 0}, (f_{-i}, Q')] <
\mathbb{E}[T_{i}|\vec{h}_{i, 0}, (f_{-i}, Q)]$ which implies
$\mathbb{E}[T_{i}|\vec{h}_{i, 0}, (f_{-i}, Q')] < \mathbb{E}[T_i|
\vec{h}_{i, 0}, f]$, which contradicts the assumption that
$f$ is in equilibrium.


Notice now that protocols $Q$ and $Q'$ are identical for any $ t \neq
\tau^*$, and if there are at least 3 pending players at $\tau^*$
(i.e. Alice and at least two others), then there would be a collision
at $\tau^*$ no matter which of the two protocols Alice uses (i.e. the
same players that were pending at $\tau^*$ would be pending at the
start of time slot $\tau^*+1$ as well). Therefore, the two protocols
behave the same in this case. However, if there are exactly 2 pending
players at $\tau^*$ (i.e. Alice and exactly one more, say Bob) the two
protocols behave differently. Indeed, if Alice uses protocol $Q$, then
there will be a collision at $\tau^*$, leaving exactly 2 pending
players at $\tau^*+1$. However, if Alice uses protocol $Q'$, then Bob
will be able to successfully transmit at $\tau^*$, leaving Alice the
only pending player at time $\tau^*+1$, which implies a strictly
smaller expected latency. The proof is completed by noting that, by
definition of $\tau^*$, the probability that there will be exactly 2
players pending at $\tau^*$ is strictly positive (since there are at
least $n-2$ steps before $\tau^*-1$ with transmission probability strictly less than 1). \qed

\medskip


Now we conclude with the impossibility result for backoff protocols, the proof of
which shares similarities to the proof of Corollary~\ref{cor:nonblocking}.

\begin{theorem} \label{lemma:nobackoff} There is no anonymous backoff
  protocol $f$ in equilibrium for $n \geq 2$ players with
  $\mathbb{E}[T_i| \vec{h}_{i, 0}, f] < \infty$, for any player $i$.
\end{theorem}
\proof 

Assume for the sake of contradiction that $f$ is in equilibrium and
let $\tau^* \stackrel{def}{=} \left\lfloor \mathbb{E}[T_i| \vec{h}_{i, 0}, f] \right\rfloor$ be finite, where $i$ is a fixed player using
$f$. By definition, we have that $f_i = \{p_{i, k}\}_{k \geq 0}$, where $p_{i, k}$ denotes the transmission probability of player $i$ after $k$ unsuccessful transmissions. Notice also that we may assume without loss of generality that $p_{i, 0} \neq 1$. Indeed, suppose there is finite integer $s>0$, such that $p_{i, k'} = 1$, for all $k'< s$ and $p_{i, s} \neq 1$ (if $s$ is not finite, then clearly $f$ does not have finite expected latency). Then the protocol $f' = \{p'_{i, k}\}_{k \geq 0}$, with $p'_{i, k} = p_{i, k+s}$, for all $k \geq 0$ is also an equilibrium.

Consider now the protocol $g = g(\tau^*)$ defined in equation
(\ref{eq-deviator}), where the first $\tau^*$ terms of $\pi$ are set
to 0. Clearly, any player using $g$ has expected latency at least $\tau^*+1$. 
Notice also that $\pi$ is consistent with $f$ up to $\tau^*$,
since $\Pr\{\vec{h}_{i, \tau^*} = (0, \ldots, 0) | f\} = (1-p_{i,
  0})^{\tau^*} > 0$. Therefore, by Lemma \ref{lemma:bestresponses} we
have that $\tau^*+1 > \mathbb{E}[T_i| \vec{h}_{i, 0}, f] =
\mathbb{E}[T_{i}|\vec{h}_{i, 0}, (f_{-i}, g)] \geq \tau^*+1$, which is a
contradiction. But this implies that, either $f$ is not in
equilibrium, or $\tau^*$ is $\infty$. \qed



\section{An efficient protocol in equilibrium}
\label{sec:efficientprotocol}

In this section we present a deadline protocol for $n$ players that is efficient, i.e. with high probability the latency of any player is $\Theta(n)$. Let $t_0 = t_0(n)$ be an integer, to be determined later and let $\beta \in (0, 1)$ be a fixed constant. We consider the following deadline protocol ${\cal Q}$ with deadline $t_0$, which is defined as follows: The $t_0-1$ time steps before the deadline are partitioned into $k+1$ consecutive intervals $I_1, I_2, \ldots, I_{k+1}$, where $k=k(n)$ is the unique integer satisfying $\beta^{k+1} n \leq \sqrt{n} < \beta^{k} n$. For any $j \in \{1, \ldots, k+1\}$, define $n_j = \beta^j n$. For $j \in \{1, \ldots, k\}$ the length of interval $I_j$ is $\ell_j = \left\lfloor\frac{e}{\beta} n_j \right\rfloor$. Interval $I_{k+1}$ is special and has length $\ell_{k+1}=n$.  In particular, this gives 
\begin{equation*}
t_0 \stackrel{def}{=} 1 + \sum_{j = 1}^{k+1} \ell_j \leq 1+n + en \sum_{j = 1}^{k} \beta^{j-1} = 1+n+en \frac{1 - \beta^{k-1}}{1-\beta} \leq n \left( 1 + \frac{e}{1-\beta}\right),
\end{equation*}
where the last inequality holds for any constant $\beta \in (0, 1)$ and $n \to \infty$. For any $t \geq 1$, the decision rule at time $t$ for protocol ${\cal Q}$ is given by 
\begin{equation}
{\cal Q}_t = \left\{ 
\begin{array}{ll}
	\frac{1}{n_j}, & \quad \textrm{if $t \in I_j, j = 1, 2 , \ldots, k+1$} \\
	1, & \quad \textrm{if $t \geq t_0$.}
\end{array}
\right.
\end{equation}

Notice that, by definition, ${\cal Q}$ is an age-based protocol. Furthermore, if at least two out of $n$ players use protocol ${\cal Q}$, then, no matter what protocol the rest of the players use, there is a non-zero probability that there will be no successful transmission until the deadline $t_0$, and thus all players will remain pending for ever. In particular, this is at least the probability that the two players using ${\cal Q}$ attempt a transmission in every step until $t_0$, which happens with probability $\prod_{t=1}^{t_0-1} ({\cal Q}_t)^2 \geq \frac{1}{n^{t_0}}>0$. Therefore, if there are at least two players using ${\cal Q}$, the expected latency of any player is $\infty$, hence ${\cal Q}$ is in equilibrium, for any $n \geq 3$ and deadline $t_0$.

In Theorem \ref{thm:efficiency} we prove that ${\cal Q}$ is also efficient; when all players in the system use protocol ${\cal Q}$, then with high probability all players will successfully transmit before the deadline $t_0$. For the proof, we use two elementary Lemmas that formalize the fact that, in each interval, a significant number of players successfully transmit with high probability. 
For the proofs, we employ standard concentration results from probability theory. Full details can be found in Appendix \ref{appendix:lem:deadlineInterval0k} and \ref{appendix:lem:deadlineIntervalk+1}.

\begin{lemma} \label{lem:deadlineInterval0k}
Assume that all players in the system use protocol ${\cal Q}$. For any $j\in\{1, \ldots, k\}$, if the number of pending players before interval $I_j$ is at most $n_j$, then after $I_{j}$, with probability at least $1-\exp(-\frac{1}{3} \beta^{j+2} n)$ there will be at most $n_{j+1}$ pending players.	
\end{lemma}

\begin{lemma} \label{lem:deadlineIntervalk+1}
If the number of pending players at the start of interval $I_{k+1}$ is at most $n_{k+1}$, then after interval $I_{k+1}$, with probability at least $1-\exp\left( - \frac{1}{3} n_{k+1}\right)$ all players will have successfully transmitted.	
\end{lemma}

We are now ready to prove our main Theorem.

\begin{theorem} \label{thm:efficiency}
Protocol ${\cal Q}$ is efficient. In particular, for any constant $\beta \in (0, 1)$, when all players use ${\cal Q}$, the probability that there is a pending player after time $t_0 \leq n \left( 1 + \frac{e}{1-\beta}\right)$ is at most $\exp(-\Theta(\sqrt{n}))$.
\end{theorem}
\proof It suffices to show that with high probability every player will have successfully transmitted before $t_0$. Note that, the probability that there are still pending players at $t_0 = \Theta(n)$ is upper bounded by the probability that (a) there exists $j \in \{1, 2, \ldots, k\}$ such that, at the end of interval $I_j$ there are more than $n_{j+1}$ pending players, or (b) there are still pending players after interval $I_{k+1}$.

Therefore, by Lemma~\ref{lem:deadlineInterval0k} and Lemma~\ref{lem:deadlineIntervalk+1} and the union bound, the probability that not all players successfully transmit before $t_0$ is at most 
\begin{equation}
\exp\left( - \frac{1}{3} n_{k+1}\right) + \sum_{j=1}^k \exp\left(-\frac{1}{3}\beta^2 n_j\right).
\end{equation}
Since $n_{j} \geq n_{k+1} \geq \beta \sqrt{n}$, for any $j \in \{1, 2, \ldots, k\}$, the above upper bound becomes $(k+1) \exp\left( - \Theta(\sqrt{n})\right)$. The proof is concluded by noting that, by definition of $k$, we have $k = \Theta(\log{n})$. \qed

\medskip

We note that, in our analysis, $\beta \in (0, 1)$ can be any constant arbitrarily close to 0, therefore, by Theorem \ref{thm:efficiency}, the upper bound on the latency of protocol ${\cal Q}$ can be as small as $(1+e)n+o(n)$ with high probability.

\newpage





\bibliography{aloha}




\appendix


\section{Proof of Claim \ref{claim:protocol}}
\label{appendix:claim:protocol}
  The situation from Alice's perspective can be modeled as a Markov
  chain ${\cal M}$ with state space $\{A, B, C, D\}$. $A$ is the
  initial state where both players are pending (and they both know
  this). $A$ is reached either in $t=1$, or when Alice transmitted in
  the previous time step and there was a collision. State $B$ models
  the case when both players are pending, but Alice does not know
  this, because she did not transmit in the previous time step. State
  $C$ is reached when only Alice is pending; notice that, by
  definition of the protocol, there is no way for Alice to distinguish with
  certainty between states $B$ and $C$ if both herself and Bob use
  $f$. Finally, $D$ is the state in which Alice has successfully
  transmitted. The transition graph of ${\cal M}$ is shown in Figure
  \ref{MC-both-f}.

For example, we can see from the transition graph that the probability
that we visit state $A$ at time $t+1$, given that we are in $B$ at $t$
is given by $\Pr({\cal M}_{t+1} = A| {\cal M}_t = B) = 1$. Indeed, if
${\cal M}_t = B$, neither player transmitted at $t$, so both will
transmit with probability 1 at $t+1$, causing a collision, after which
Alice (and also Bob) can deduce that all players are still
pending. Similarly, $\Pr({\cal M}_{t+1} = D| {\cal M}_t = C) = 1$,
because, being at $C$ means that only Bob transmitted (successfully)
at $t$ and so Alice will transmit (also successfully, being the only
pending player) at $t+1$ with probability 1.

Clearly, $C_{Alice}^{f,1}$ is equal to the expected hitting time
$k_{A}^{D}$ that ${\cal M}$ needs to reach state $D$, given that we
start from $A$. By definition, we have $k_C^D=1$, $k_D^D=0$, and by
the Markov property, we get $k_{B}^D=k_A^D+1$ and $k_A^D = 1 +
\frac{4}{9} k_A^D + \frac{1}{9} k_B^D + \frac{2}{9} k_C^D +
\frac{2}{9} k_D^D$. By rearranging terms and making the substitutions,
we conclude that $C_{Alice}^{f,1} = 3$. 

Calculating $C_{Alice}^{f, 0}$ is a bit more tricky, because
since Alice did not attempt transmission at the previous slot, she
cannot be certain in which state she is, but she knows that is at
state $B$ with probability $1/3$ and in $C$ with $2/3$. Therefore 
$C_{Alice}^{f,0}=\frac{1}{3}k_B^D+\frac{2}{3}k_C^D=2$. \qed


\section{Proof of Corollary \ref{cor:nonblocking}}
\label{appendix:cor:nonblocking}

Assume for the sake of contradiction that $f$ is in equilibrium
and let $\tau^* \stackrel{def}{=} \left\lfloor \mathbb{E}[T_i|
  \vec{h}_{i, 0}, f] \right\rfloor$ be finite, where $i$ is a fixed
player using $f$. Consider the protocol $g = g(\tau^*)$ as defined in
(\ref{eq-deviator}), where the first $\tau^*$ terms of $\pi$ are set
equal to 0. Clearly, any player using $g$ has expected latency
at least $\tau^*+1$, irrespectively of the transmissions of
the other players. Notice also that $\pi$ is consistent with $f$ up to
$\tau^*$, since $\Pr\{\vec{h}_{i, \tau^*} = (0, \ldots, 0) | f\} =
\prod_{t=1}^{\tau^*} (1-p_{i, t}) > 0$. Therefore, by Lemma
\ref{lemma:bestresponses} we have that $\tau^*+1 > \mathbb{E}[T_i|
\vec{h}_{i, 0}, f] = \mathbb{E}[T_{i}|\vec{h}_{i, 0}, (f_{-i}, g)] \geq
\tau^*+1$, which is a contradiction. We conclude that either $f$ is not
in equilibrium, or $\tau^*$ is $\infty$. \qed


\section{Proof of Claim \ref{claim:agebased}}
\label{appendix:claim:agebased}

For any time
$t$, define $Z^{f}_{t}$ to be the number of non-blocking probabilities
of the protocol $f$ up to $t$, i.e. $Z^{f}_{t} \stackrel{def}{=}
\sum_{t' \leq t} (1 - \left\lfloor f_{t'} \right\rfloor)$. Set $\tau'
\stackrel{def}{=} \inf\{t: f_t=1, Z^{f}_{t} \geq n-1 \}$. Assume for
the sake of contradiction that there does not exist a $\tau^*$ with
the properties described in the claim. In particular, this means that
$\tau' = \infty$. However, the latter can happen if one of the
following cases is true:

\begin{description}

\item[(i)] There is no finite $\tau$ such that $f_{\tau} = 1$.

\item[(ii)] There exists finite $\tau$ such that $f_{\tau} = 1$, $Z^{f}_{t} \leq n-2$ and $f_t = 1$, for all $t \geq \tau$.

\item [(iii)] There exists finite $\tau$ such that $f_{\tau} = 1$,
  $Z^{f}_{t} \leq n-2$ and $f_t < 1$, for all $t \geq \tau$.

\end{description}  

We now prove that in all those cases we get a contradiction. Case (i)
comes in contradiction with Corollary \ref{cor:nonblocking}. 

If case (ii) holds, then clearly, if all players use $f$, at most
$n-2$ players can successfully transmit before $\tau$ and the rest
will remain pending for ever. But this means that the expected latency
of a player $i$ using $f$ is at least
\begin{equation*}
\Pr\{\textrm{$i$ does not successfully transmit before $\tau$}| \vec{h}_{i, 0}, f\} \cdot \infty = \infty,
\end{equation*}
which leads to a contradiction, since we assumed $\mathbb{E}[T_i| \vec{h}_{i, 0},
f] < \infty$.

Suppose now that case (iii) holds. Consider the protocol $g$ defined
as follows:

\begin{equation} 
g \stackrel{def}{=} \left\{
\begin{array}{ll}
	0, & \quad \textrm{if $f_t < 1$, for $1 \leq t \leq \mathbb{E}[T_i| \vec{h}_{i, 0}, f]$} \\
	1, & \quad \textrm{if $f_t = 1$, for $1 \leq t \leq \mathbb{E}[T_i| \vec{h}_{i, 0}, f]$} \\
	f_t, & \quad \textrm{for $t > \mathbb{E}[T_i| \vec{h}_{i, 0}, f]$.}
\end{array}
 \right.
\end{equation}
Let $i$ be a fixed player (say Alice). Notice that, if all other
players use $f$ and Alice uses $g$, then Alice has expected latency
strictly larger than $\mathbb{E}[T_i| \vec{h}_{i, 0}, f]$; indeed, for
any $t \leq \mathbb{E}[T_i| \vec{h}_{i, 0}, f]$, Alice only attempts a
transmission when $f_t = 1$ and there is at least one more other
pending player using $f$, and so there is a collision. However, since the initial (deterministic) sequence of $\left\lfloor\mathbb{E}[T_i| \vec{h}_{i, 0}, f] \right\rfloor$ transmissions of 
$g$ is consistent with $f$, by Lemma
\ref{lemma:bestresponses} we have that $\mathbb{E}[T_i| \vec{h}_{i,
  0}, f] = \mathbb{E}[T_{i}|\vec{h}_{i, 0}, (f_{-i}, g)] >
\mathbb{E}[T_i| \vec{h}_{i, 0}, f]$, which is a contradiction. This
completes the proof of the claim. \qed



\section{Proof of Lemma \ref{lem:deadlineInterval0k}}
\label{appendix:lem:deadlineInterval0k}

Fix $j \in \{1, \ldots, k\}$ and assume that the precondition of the lemma is fulfilled, i.e., before interval $I_j$ there are at most $n_{j}$ pending players. 
Let $r_t$ denote the number of pending players at time $t$. In particular, for any $t \in I_j$, if the preconditions of the lemma is fulfilled, we have $r_t\leq n_j$. Therefore the probability of a successful transmission in round $t \in I_j$ is given by

\begin{equation*}
r_t {\cal Q}_t (1-{\cal Q}_t)^{r_t-1}	\geq r_t {\cal Q}_t (1-{\cal Q}_t)^{n_j-1} = r_t \frac{1}{n_j} \left( 1 - \frac{1}{n_j}\right)^{n_j-1} \geq \frac{1}{e}  \frac{r_t}{n_j},
\end{equation*}
where in the last inequality we used the fact that $\left(1-\frac{1}{x}\right)^{x-1} \geq \frac{1}{e}$, for any $x > 1$. Therefore, for any round $t\in I_j$, either we already have $r_t\leq n_{j+1}=\beta n_j$ pending players, or the probability of a successful transmission in round $t$ is at least $a \stackrel{def}{=} \frac{1}{e} \frac{n_{j+1}}{n_j} = \frac{\beta}{e}$.

Let now $X_j$ be the random variable counting the number of successful transmissions in interval $I_j$. Notice that, by the above discussion, given that at the start of interval $I_j$ there are at least $n_{j+1}$ pending players, $X_j$ stochastically dominates a Binomial random variable $Y_j \sim Bin(\ell_j, a)$, with mean value $\ell_j \cdot a$. Therefore, by a Chernoff bound (see \cite{RossIntrobook}), we get 

\begin{equation*}
\Pr(X_j<(1-\beta) \ell_j \cdot a) \leq \Pr(Y_j<(1-\beta) \ell_j \cdot a) \leq \exp\left(-\frac{1}{2}\beta^2 \ell_j \cdot a\right) \leq \exp\left(-\frac{1}{3}\beta^2 n_j\right),
\end{equation*}
where in the last inequality we used the fact that, by definition, $n_j \geq \sqrt{n}$, for all $j \leq k$, thus $\ell_j \cdot \frac{\beta}{e} \geq \frac{2}{3} n_j$. This directly implies the lemma. \qed


\section{Proof of Lemma \ref{lem:deadlineIntervalk+1}}
\label{appendix:lem:deadlineIntervalk+1}

Consider a fixed player (say Alice) that is pending at the start of interval $I_{k+1}$. Given that there are at most $n_{k+1} = \beta^{k+1} n$ pending players at any time step $t \in I_{k+1}$, the probability that Alice successfully transmits during $t$ is at least 

\begin{equation*}
{\cal Q}_t (1 - {\cal Q}_t)^{n_{k+1}-1} = \frac{1}{n_{k+1}} \left( 1 - \frac{1}{n_{k+1}}\right)^{n_{k+1}-1}.
\end{equation*} 
Therefore, since $|I_{k+1}| = \ell_{k+1} = n$, the probability that Alice is still pending after interval $I_{k+1}$ is at most 

\begin{equation} \label{eq:failureprobability}
\left( 1 -\frac{1}{n_{k+1}} \left( 1 - \frac{1}{n_{k+1}}\right)^{n_{k+1}-1}\right)^{n} \leq \exp\left( - \frac{n}{n_{k+1}} \left( 1 - \frac{1}{n_{k+1}}\right)^{n_{k+1}-1} \right).
\end{equation} 

Recall that, by definition, $k$ is the (unique) smallest integer satisfying $n_{k+1} \leq \sqrt{n} < n_k$. In particular, this implies that $n_{k+1} > \beta \sqrt{n}$, therefore $n_{k+1}$ goes to $\infty$ as $n \to \infty$. Additionally, we have that $\frac{n}{n_{k+1}} \geq n_{k+1}$. Therefore, using the fact that $\left(1-\frac{1}{x}\right)^{x-1} \geq \frac{1}{e}$, for any $x > 1$, the right hand side of (\ref{eq:failureprobability}) is at most $\exp\left( -\frac{1}{e} n_{k+1} \right)$. 

By the union bound, given that there are at most $n_{k+1}$ pending players at the start of interval $I_{k+1}$, the probability that there is at least one pending player after $I_{k+1}$ is at most $n_{k+1} \exp\left( -\frac{1}{e} n_{k+1} \right) \leq \exp\left( - \frac{1}{3} n_{k+1}\right)$, as stated in the Lemma. \qed

\end{document}